\documentclass{article}

\usepackage{arxiv}

\usepackage[pagebackref=true]{hyperref}

\usepackage{graphics} 
\usepackage{epsfig} 
\usepackage{mathptmx} 
\usepackage{times} 
\usepackage{amsmath} 
\usepackage{amssymb}  
\usepackage{amsthm}
\usepackage{xypic}
\usepackage{accents}
\usepackage[ruled,lined]{algorithm2e}

\usepackage[dvipsnames]{xcolor}

\usepackage{xcolor}  

\usepackage{amssymb}

\usepackage{amsmath}
\usepackage{mathdots}

\usepackage{mathtools} 

\usepackage{tensor}

\usepackage{urwchancal}
\DeclareFontFamily{OT1}{pzc}{}
\DeclareFontShape{OT1}{pzc}{m}{it}{<-> s * [1.10] pzcmi7t}{}
\DeclareMathAlphabet{\mathpzc}{OT1}{pzc}{m}{it}

\usepackage{graphicx}
\usepackage{ifpdf}
\ifpdf
\usepackage{epstopdf}
\PrependGraphicsExtensions{.svg}
\fi

\newtheorem{theorem}{Theorem}[section]
\newtheorem{lemma}[theorem]{Lemma}

\newtheorem{remark}[theorem]{Remark}

\providecommand{\Order}{\mathbb{O}}

\providecommand{\R}{\mathbb{R}}

\providecommand{\SO}{\mathbf{SO}}

\providecommand{\gothX}{\mathfrak{X}} 

\providecommand{\calM}{\mathcal{M}}
\providecommand{\calN}{\mathcal{N}}

\providecommand{\vecL}{\mathbb{L}}

\providecommand{\tT}{\mathrm{T}} 

\providecommand{\GP}{\mathbf{N}} 

\DeclareMathOperator{\diag}{diag}

\providecommand{\PT}{\mathbf{P}} 

\providecommand{\td}{\mathrm{d}}
\providecommand{\tD}{\mathrm{D}}

\usepackage{accents}
\usepackage{mathtools}
\makeatletter
\providecommand{\scirc}{%
    \hbox{\fontfamily{\rmdefault}\fontsize{0.4\dimexpr(\f@size pt)}{0}\selectfont{\raisebox{-0.52ex}[0ex][-0.52ex]{$\circ$}}}}

\makeatother

\makeatletter
\providecommand{\ucirc}{%
    \hbox{\fontfamily{\rmdefault}\fontsize{0.4\dimexpr(\f@size pt)}{0}\selectfont{\raisebox{0.0ex}[0ex][-0.52ex]{$\circ$}}}}

\makeatother

\mathchardef\mhyphen="2D

\providecommand{\idx}[4]{\tensor*[_{#3}^{#2}]{#1}{_{#4}}}

\providecommand{\ids}[5]{\tensor*[_{#3}^{#2}]{#1}{^{#5}_{#4}}}

\providecommand{\Order}{\mathbf{O}} 

\providecommand{\etal}{\textit{et al.}~}

\usepackage{stfloats}

\newcommand{\papercitation}{
Ge, Y., van Goor, P., Mahony, R. (2023), ``A Note on the Extended Kalman Filter on a Manifold'', \emph{Accepted for presentation at IEEE CDC}.
}

\newcommand{\publicationdetails}{
\papercitation \\
\copyrightNoticeIEEE{2023} 
}
\newcommand{\publicationversion}
{Author accepted version}

\begin{document}
\title{
A Note on the Extended Kalman Filter on a Manifold
}
\headertitle{
A Note on the Extended Kalman Filter on a Manifold
}
\author{
    \href{https://orcid.org/0000-0001-7969-7039}{\includegraphics[scale=0.06]{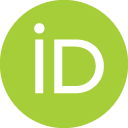}\hspace{1mm}
Yixiao Ge}
\\
    Systems Theory and Robotics Group \\
    School of Engineering \\
	Australian National University \\
    ACT, 2601, Australia \\
    \texttt{Yixiao.Ge@anu.edu.au} \\
\And    \href{https://orcid.org/0000-0003-4391-7014}{\includegraphics[scale=0.06]{orcid.png}\hspace{1mm}
Pieter van Goor}
\\
    Systems Theory and Robotics Group \\
    School of Engineering \\
	Australian National University \\
    ACT, 2601, Australia \\
    \texttt{Pieter.vanGoor@anu.edu.au} \\
	\And	\href{https://orcid.org/0000-0002-7803-2868}{\includegraphics[scale=0.06]{orcid.png}\hspace{1mm}
    Robert Mahony}
\\
    Systems Theory and Robotics Group \\
    School of Engineering \\
	Australian National University \\
    ACT, 2601, Australia \\
	\texttt{Robert.Mahony@anu.edu.au} \\
}

\maketitle

\begin{abstract}
The kinematics of many control systems, especially in the robotics field, naturally live on smooth manifolds.
Most classical state-estimation algorithms, including the extended Kalman filter, are posed on Euclidean space.
Although any filter algorithm can be adapted to a manifold setting by implementing it in local coordinates and ignoring the geometric structure, it has always been clear that there would be advantages in taking the geometric structure into consideration in developing the algorithm.
In this paper, we argue that the minimum geometric structure required to adapt the extended Kalman filter to a manifold is that of an affine connection.
With this structure, we show that a naive coordinate implementation of the EKF fails to account for geometry of the manifold in the update step and in the reset step.
We provide geometric modifications to the classical EKF based on parallel transport of the measurement covariance (for the update) and \emph{a-posteriori} state covariance (for the reset) that address these limitations.
Preliminary results for attitude estimation with two directional measurements demonstrate that the proposed modifications significantly improve the transient behavior of the filter.
\end{abstract}

\section{INTRODUCTION}
Over the past sixty years, the extended Kalman filter (EKF) has been the industry standard for state estimation problems when the system dynamics are governed by nonlinear equations \cite{kalman1960new}\cite{smith1962application}.
Although the classical formulation of EKF is given in global Euclidean space, there have been many papers that adapt the Kalman filter methodology to systems that live on smooth manifolds.
Although any choice of local coordinates provides a representation in which the classical extended Kalman filter can be implemented, from as early as the 1970s, authors were demonstrating the advantage of choosing local coordinate charts that encode geometric structure, such as Riemannian metrics or homogeneous symmetries, in a natural manner \cite{duncan1977some}\cite{ng1985nonlinear}.
Such charts have specific structure that can lead to lower linearisation error for the filter improving performance.
The approach comes at the cost of using a new chart for every iteration of the algorithm, since the nice geometric properties of the chart usually only hold at a single point, the origin or reference point of the chart.
This structure is also the foundation of the $\boxplus$ ('boxplus') and $\boxminus$ ('boxminus') operators used for modelling state displacement on manifold, introduced in \cite{hertzberg2013integrating}.
Similar techniques were used in \cite{clemens2016extended} and \cite{brossard2020code}.
One example of local coordinates centered at a point are the normal coordinates comprising geodesics in a star shaped neighbourhood associated with an affine connection \cite{ng1985nonlinear}.
This is particularly of interest for Riemannian manifolds where Levi-Civita connection  is the unique torsion free connection that preserves the metric \cite{duncan1977some}\cite{hauberg2013unscented}\cite{said2013filtering}.
However, normal coordinates can be defined for an arbitrary affine connection and do not require a Riemannian metric, for example, one-parameter subgroups on a Lie-group state-space.

A parallel, and closely related, research thread considers geometric structure induced by global symmetry properties of the system state-space.
A key early problem that motivated this perspective was the question of attitude estimation in aerospace applications.
The use of quaternions or rotation matrices for orientation representation led to the multiplicative extended Kalman filter (MEKF) \cite{lefferts1982kalman}.
This methodology uses an extended Kalman filter based on a linearisation of a global error defined using the group structure of the state-space.
Although the MEKF became an industry standard for attitude estimation in aerospace \cite{markley2014fundamentals} the approach was not extended beyond the quaternion and rotation groups until picked up again after the turn of the century.
The advent of remotely piloted aerial vehicle led to a revaluation of the attitude filtering problem \cite{Thienel_2003}\cite{Mahony_2008}.
Bonnabel \etal proposed a general theory for the Invariant Extended Kalman Filter (IEKF) for systems on Lie-groups in a series of works \cite{Bonnabel_2007}\cite{Bonnabel_2009}, motivated originally by attitude estimation.
They identified a class of `group affine' systems for which they showed that the IEKF provides an exact linearisation of the prediction step of the EKF  \cite{Barrau_2017} leading to global convergence and high performance.
This observation is also seen in the work by Long \etal \cite{long2013banana}.
Bourmaud \etal \cite{bourmaud2013discrete}\cite{Bourmaud_2015} proposed the continuous-discrete EKF (CD-EKF) for systems on connected unimodular Lie groups using concentrated Gaussian distribution.
Mahony \etal \cite{Mahony_2022}\cite{van2022equivariant}\cite{ge2022equivariant} proposed the equivariant filter (EqF), a general Kalman filter design methodology for systems on homogeneous spaces.
This design perspective has been widely adapted by modern filters and applied to a range of real-world problems including visual inertial odometry \cite{loianno2016visual}\cite{van2022eqvio}, inertial navigation systems \cite{hashim2021geometric}\cite{fornasier2022equivariant}, homography tracking \cite{hamel2011homography}, etc.
The design approach can be interpreted as using local coordinates induced by one of the Cartan-Schouten affine connections on the Lie-group.
These are the affine connections for which the geodesics are the one-parameter subgroups, or exponentials, on the Lie-group (or their projections on the homogeneous space).
Although this insight draws a strong analogy with the more general research theme discussed earlier, the Lie-group and homogeneous space structure fundamental in this approach (for common real-world problems) carries global matrix algebraic structure that is core to the mathematical formulations and efficiency of the filters.

In this paper, we consider the nonlinear systems where the state space and output space are both smooth manifolds that admit an affine connection and present a general error-state extended Kalman filter design methodology for such systems.
An affine connection is the minimum geometric structure that captures the key properties of geometry.
We argue that this is the minimum requirement on manifold structure that is necessary to consider to develop EKF algorithms that are not simply an implementation of the classical Euclidean EKF in local coordinates.
With an affine connection, one can define the concepts of geodesic and parallel transport that we show are fundamental in deriving high performance Kalman filter algorithms.
We frame the filter development in the context of using concentrated Gaussian distributions \cite{Bourmaud_2015}\cite{ge2022equivariant} as local approximations of the information state of the filter.
Using this formalisation we can define concepts of mean and covariance for the approximate information state in normal coordinates on the manifold without requiring these concepts to be well defined for the full information state on the manifold.
We go on to show that the geometric structure of the manifold should be taken into account explicitly in two places in a Kalman filter update.
The first in the Bayesian update step to parallel transport the covariance of the generative noise measurement model into coordinates adapted to the state.
Then second in the reset step to parallel transport the state covariance estimate generated by the Bayesian update to the new filter estimate coordinates.
There are several prior works that have considered geometric modification to the reset step for EKF algorithms on manifolds \cite{loianno2016visual}\cite{mueller2017covariance}\cite{gill2020full}\cite{Mahony_2022}\cite{ge2022equivariant}, however, the authors are not aware of any prior work in the EKF algorithms for geometric modifications of the update step.
In addition, we believe the geometric insight in the present paper to be novel.
We provide simulations to demonstrate the advantage of the proposed geometric modifications during the transient response of the filter.
This is the period of the filter evolution when not taking account the geometry of the manifold impacts most on the filter performance.
Once the filter has converged to steady-state tracking, the local linearisation error is small and all EKF algorithms have similar performance.

This paper includes seven sections alongside the introduction and conclusion.
In Section \ref{sec:notation} notation is defined and preliminary mathematics is discussed.
In Section \ref{sec:formulation} the system and its stochastic model are defined.
In Section \ref{sec:algorithm}, we present a conventional error-state EKF construction based on normal coordinates on smooth manifolds.
In Section \ref{sec:geometric}, we propose novel geometric modifications in the filter dynamics to compensate the coordinate transform during the measurement update step and the error reset step.
The example of attitude estimation with directional measurements is presented in Section \ref{sec:sim}.

\section{PRELIMINARIES}
\label{sec:notation}

\subsection{Manifold and Affine Connection}

Let $\calM$ be a smooth manifold with dimension $m$.
The tangent space at a point $\xi\in\calM$ is denoted $\mathrm{T}_\xi\calM$.
The tangent bundle is denoted $\tT\calM$.
Given a differentiable function between smooth manifolds $h:\calM\rightarrow\calN$, its derivative at $\xi^\circ$ is written as
\begin{align*}
    \tD_{\xi|\xi^\circ}h(\xi): \tT_{\xi^\circ}\calM\rightarrow \tT_{h(\xi^\circ)}\calN.
\end{align*}
The notation $\tD h(\xi):\tT\calM\rightarrow \tT\calN$ denotes the differential of $h$ with an implicit base point.

Let $\gothX(\calM)$ denote the space of differentiable vector fields over $\calM$.
Let $C^\infty(\calM)$ denote the class of infinitely differentiable functions on $\calM$.
An affine connection \cite{lee2018introduction} is an operator $\nabla:\gothX(\calM)\times\gothX(\calM)\rightarrow\gothX(\calM)$ written $(X,Y)\mapsto\nabla_X Y$, that satisfies
\begin{itemize}
    \item $(Linear\;in\;X)\quad\nabla_{f_1 X_1+f_2 X_2} Y = f_1 \nabla_{X_1}Y+f_2 \nabla_{X_2}Y $,
    \item $(Linear\;in\;Y)\quad\nabla_X (a_1 Y_1+a_2 Y_2) = a_1 \nabla_X Y_1+a_2 \nabla_X Y_2$,
    \item $(Product\;rule)\quad\nabla_X(fY)=f\nabla_X Y+(Xf)Y$,
\end{itemize}
for all $f_1, f_2\in C^\infty(\calM)$, $a_1, a_2\in\R$ and $X,Y\in\gothX(\calM)$.
This gives a notion of directional derivative of a vector field defined on the manifold.

A curve $\gamma:I\rightarrow\calM$ is called geodesic \cite{lee2018introduction} if
\begin{align}
   \nabla_{\dot{\gamma}(t)}\dot{\gamma}(t) = 0
\end{align}
for any $t\in I$ where $I$ is a maximal open interval in $\R$ containing 0.
For any $\hat{\xi} \in \calM$ and $v \in \tT_{\hat{\xi}}\calM$, there exists a unique maximal geodesic $\gamma:[0,t(v,\hat{\xi})) \rightarrow \calM$ that satisfies $\gamma_v(0)=\hat{\xi}$ and $\dot{\gamma}_v(0) = v$ \cite[Collorary 4.28]{lee2018introduction}.
The exponential mapping $\exp_{\hat{\xi}}: W_{\hat{\xi}} \subset \tT_{\hat{\xi}}\calM\rightarrow\calM$ is defined as mapping each tangent vector $v \in \tT_{\hat{\xi}} \calM$ to the value of its geodesic at time 1; that is,
\begin{align}
   \exp_{\hat{\xi}}(v) = \gamma_v(1)
\end{align}
where $W_{\hat{\xi}}$ is the largest open subset of $\tT_{\hat{\xi}}$ for which $\exp$ is a diffeomorphism.
Let $U_{\hat{\xi}} = \exp_{\hat{\xi}} (W_{\hat{\xi}})$ and note that $U_{\hat{\xi}}$ is open by construction.
Let $\imath_{\hat{\xi}} : \tT_{\hat{\xi}} \calM \to \R^m$ provide a linear isomorphism between $\tT_{\hat{\xi}} \calM$ and $\R^m$ for each $\hat{\xi}$.
Then the \emph{normal} coordinates on $\calM$ are defined by
\begin{align}
\vartheta_{\hat{\xi}}:= \imath_{\hat{\xi}} \circ \exp_{\hat{\xi}}^{-1} : U_{\hat{\xi}} \to \R^m.
\label{eq:vartheta}
\end{align}

A vector field is an assignment $\xi \mapsto X_\xi$ for every $\xi \in \calM$.
A vector field $X$ is \emph{parallel} along $\gamma$ with respect to the connection $\nabla$ if
\begin{align}
   \nabla_{\dot{\gamma}(t)} X_{\gamma(t)} = 0
\label{eq:parallel_transport}
\end{align}
for all $t$.
Given any vector $X_{\gamma(0)} \in \tT_{\gamma(0)} \calM$ and smooth curve $\gamma(t)$ then there is a unique family of vectors $X_{\gamma(t)}$ that satisfy \eqref{eq:parallel_transport} and this correspondence induces an invertible linear  map $\PT_{\gamma(t)} : \tT_{\gamma(0)} \calM \to \tT_{\gamma(t)} \calM$ between tangent spaces  by
\[
\PT_{\gamma(t)} X_{\gamma(0)} := X_{\gamma(t)}
\]
for all $X_{\gamma(0)} \in \tT_{\gamma(0)} \calM$.
For a fixed time $T$, the inverse of the parallel transport $\PT_{\gamma(T)}$ along a curve $\gamma(t)$ is equal to the parallel transport $\PT_{\gamma'(T)}$ along the reversed curve $\gamma'(t) := \gamma(T-t)$.
Parallel transport of vector fields induces parallel transport of tensor operators $\Sigma_{\gamma(0)} : \tT_{\gamma(0)} \calM \times 
\tT_{\gamma(0)} \calM \to \R$ by 
\begin{align*} 
\PT_{\gamma(t)} \Sigma_{\gamma(0)} (X_{\gamma(t)} ,Y_{\gamma(t)}) 
& :=
\Sigma_{\gamma(0)} (\PT_{\gamma(t)}^{-1} X_{\gamma(t)} ,\PT_{\gamma(t)}^{-1}Y_{\gamma(t)}).
\end{align*}

\subsection{The $\boxplus/\boxminus$ Operators}

Building on the established literature we will use the $\boxplus$ and $\boxminus$ operator notation introduced in \cite{aufframework,hertzberg2013integrating} to model small $\R^m$ `perturbations acting on $\calM$.
Recalling the normal coordinates \eqref{eq:vartheta}, define the \emph{box plus} $\boxplus: \calM \times \R^m \rightarrow\calM$  and \emph{box minus} $\boxminus:\calM\times\calM\rightarrow\R^m$ operators
by
\begin{align}
   &\xi \boxplus u = \vartheta_\xi^{-1} (u),\\
   &\zeta \boxminus \xi = \vartheta_{\xi}(\zeta),
\end{align}
for all $\xi \in\calM$, $\zeta \in W_\xi$ and $u \in \imath(U_\xi)$.
Both the $\boxplus$ and $\boxminus$ operators are associated with geodesic curves on the manifold and in this sense are the natural generalisation of straight lines on Euclidean space.
We believe that this is the most natural geometric definition of these operators.
Any other definition introduces local coordinates that are not adapted in the natural sense to the geometry of the manifold.

It is straightforward to verify that these proposed operators satisfy
\begin{subequations} \label{eq:boxplus_axioms}
   \begin{align}
      \xi \boxplus 0 &= \xi, \\
      \xi \boxplus (\zeta \boxminus \xi) &= \zeta, \\
      (\xi \boxplus u) \boxminus \xi &= u,
   \end{align}
\end{subequations}
the first three of the four requirements of the original definition proposed in \cite[Def.~1]{hertzberg2013integrating}.
The fourth axiom in \cite[Def.~1]{hertzberg2013integrating} requires that
\begin{align} 
   |(\hat{\xi} \boxplus \delta_1) \boxminus (\hat{\xi} \boxplus \delta_2)|^2 \leq |\delta_1 - \delta_2|^2.
   \label{eq:FourthAxiom}
\end{align}
However, in the normal coordinates (on a Riemannian manifold) one has
\[
|(\hat{\xi} \boxplus \delta_1) \boxminus (\hat{\xi} \boxplus \delta_2)|^2
=
|\delta_1 - \delta_2|^2
-\frac{1}{3} \text{Ric}(\delta_1, \delta_2) + \Order(|\delta|^3),
\]
where $\text{Ric}$ is the Ricci curvature tensor.
For manifolds with non-negative curvature; that is, $\text{Ric} \geq 0$ is positive semi-definite, then \eqref{eq:FourthAxiom} holds. 
For any manifold with negative curvature, however, axiom 4 from  \cite[Def.~1]{hertzberg2013integrating} will fail locally.
This axiom was used in \cite{hertzberg2013integrating} to prove properties of the mean of the true information state on the manifold $\calM$ associated with properties of the mean of distributions defined in the $\R^m$ coordinates.
In general manifolds the concept of mean and covariance are unclear and instead we will work entirely with concentrated Gaussian distribution approximations (\S~\ref{sub:stochastic_model}) of the true information state.
These approximations do have well-defined mean and covariance parametrisation that can be used in the filter algorithm.
The results developed in this paper are general and not restricted to manifolds with non-negative curvature.

\section{PROBLEM FORMULATION}
\label{sec:formulation}

\subsection{Stochastic model}\label{sub:stochastic_model}

%

It is always possible to define a volume measure on a general manifold using a partition of unity construction.
The class of probability distributions for the information state considered will be those that are integrable with respect to such a measure.
Even with a well-defined concept of probability distribution, concepts such as mean and covariance are not well defined on a general manifold. 
There are many works that use geometric structure of the manifold such as a Riemannian metric \cite{duncan1977some}\cite{ng1985nonlinear}, or Lie-group structures \cite{Bourmaud_2015}\cite{wang2006error} to define equivalent concepts.
These constructions, however, are not necessary for the formulation of Kalman filter algorithms.
Rather such algorithms need only a definition of a class of \emph{approximating} distributions that can be parameterised by mean and covariance \emph{parameters}. 
In particular, it is not necessary that the mean and covariance parameters used as state in the filter correspond to statistics of the true distribution or even of the approximate distribution, only that the distribution generated by the filter parameterisations is close in some sense to the true distribution.

In the remainder of the paper we assume that both the systems state-space $\calM$ and the output space $\calN$ admit affine connections and that we work with normal coordinates
\begin{align}
\vartheta_{\hat{\xi}} &: \calM \rightarrow \R^m \\
\varphi_{\hat{y}}  & : \calN \rightarrow \R^n
\end{align}
We approximate a general distribution $p : \calM \to \R_+$ around $\hat{\xi} \in \calM$ by a \emph{concentrated Gaussian distribution} \cite{wang2006error}
\begin{align}
\GP_{\hat{\xi}}(\xi | \mu,\Sigma) := \alpha
\exp(-\frac{1}{2}(\vartheta_{\hat{\xi}}(\xi)-\mu)^\top\Sigma^{-1}(\vartheta_{\hat{\xi}}(\xi)-\mu)),
\end{align}
where
\[
\alpha := \left| \int_{U_{\hat{\xi}}}
\exp(-\frac{1}{2}(\vartheta_{\hat{\xi}}(\xi)-\mu)^\top\Sigma^{-1}(\vartheta_{\hat{\xi}}(\xi)-\mu))
 \td \xi \right|^{-1}.
\]
is a normalizing factor, $\mu\in \R^m$ is a mean vector parameter and $\Sigma\in\mathbb{S}_+(m)$ is a positive-definite symmetric $m\times m$ covariance matrix parameter.
Note that the support for the distribution $\GP_{\hat{\xi}}(\xi | \mu,\Sigma)$ is contained in the open set $U_{\hat{\xi}} \subset \calM$.
Within this set, the distribution in local coordinates $x = \vartheta_{\hat{\xi}}(\xi)$ looks like a trimmed Gaussian and the first and second order statistics $\mu$ and $\Sigma$ are well defined.
Although the mean and covariance have natural interpretations as \emph{parameters} in the concentrated Gaussian, they do not correspond to the statistical mean and variance of the distribution on $\calM$, or even in $\R^m$ coordinates due to the trimmed nature of the distribution.
This does not prevent them being used to parameterise the approximate distribution and derive an extended Kalman filter.
The filter formulation is now done within the class of concentrated Gaussian distribution and its validity will depend on the validity of the approximation $p(\xi) \approx \GP_{\hat{\xi}}(\xi | \mu,\Sigma)$.
Clearly for a large range of engineering applications where the \emph{a-posteriori} state distribution is locally Gaussian in the normal coordinates this approximation will work extremely well.

\subsection{System definition}

In this work, we consider a nonlinear discrete-time system living on a smooth manifold $\calM$.
The system function is given by
\begin{align}\label{eq:system_twonoise}
\xi_{k+1} = F&(\xi_k,u_{k+1} + \kappa^I_{k+1}) \boxplus \kappa^P_{k+1}, \notag\\
&\kappa^I_{k+1}\sim\GP(0,R^I_{k+1}),\quad \kappa^P_{k+1}\sim\GP(0,R^P_{k+1}),
\end{align}
where $\xi\in\calM$ and $u\in\vecL$ are the system state and input, respectively.
These systems usually admit two types of noise, the input noise $\kappa^I_{k+1}$ and the processing noise $\kappa^P_{k+1}$, modelled as Gaussian processes on the linear input space $\vecL$ and the tangent space $\tT_{\xi_{k+1}}\calM$.
In this work, we combine these noise terms through linearisation,
\begin{align}
\xi_{k+1} = F&(\xi_k,u_{k+1}) \boxplus (\kappa^P_{k+1} + B_{k+1}\kappa^I_{k+1}),
\end{align}
where $B_{k+1} = \tD_{u| u_{k+1}} F(\xi_k, u)$ is the differential of the system function with respect to the input signal.
The resulting simplified noise model is
\begin{align}\label{eq:system}
\xi_{k+1} = F(\xi_k,u_{k+1}) \boxplus \kappa_{k+1},\quad \kappa_{k+1}\sim\GP(0,R_{k+1}),
\end{align}
where the total covariance $R_{k+1} = R^P_{k+1} + B_{k+1}R^I_{k+1}B_{k+1}^\top$ captures the combined effect of processing noise and input noise.

The configuration output
\begin{align}\label{eq:configuration}
   y_{k+1} = h(\xi_{k+1})\boxplus \nu_{k+1},\quad \nu_{k+1}\sim\GP(0,Q_{k+1}),
\end{align}
is given by a function $h: \calM\rightarrow\calN$, where $\calN$ is a smooth manifold termed the \emph{output space}.
The disturbance $\nu_{k+1}$ is modelled as a Gaussian process in the normal coordinate around $h(\xi_{k+1})$ on $\calN$.

\section{EXTENDED KALMAN FILTER ON MANIFOLD}
\label{sec:algorithm}

In this section, we provide an overview of the EKF design methodology on smooth manifolds and indicate points at which the geometric structure of the manifold is not taken into account in classical treatments.
We derive the filter in the error-state formulation, since this allows for a more geometric analysis. 
The error-state Kalman filter considers propagating the information state of an error $\epsilon_k$ between true $\xi_k$ and nominal $\hat{\xi}$ \cite{sola2017quaternion} and reconstructs the filter state $\hat{\xi}$ from the error update. 


\subsubsection{Error state}
The \emph{local} error in normal coordinates $\epsilon_k \in\tT_{\hat{\xi}_k}\calM$ is given by
\begin{align}
   \epsilon_k = \xi_k\boxminus\hat{\xi}_k,
\end{align}
where $\xi_k,\hat{\xi}_k\in\calM$ are the true and estimated system state, respectively.
In general, on systems with symmetry the error state is often defined globally using a group structure and the local error state is the local linearisation of this construction \cite{Barrau_2017}\cite{van2022equivariant}. 

\subsubsection{Stochastic approximation}
The information state for the filter at time $k$, 
that approximates the true \emph{a-posteriori} distribution of the information state, is a concentrated Gaussian
\[
\xi_{k|k}\sim\GP_{\hat{\xi}_{k|k}}(\xi| 0,\Sigma_{k|k}) 
\]
where $k|k$ indicates the state at time $k$ conditioned on information (inputs and measurements) up to and including time $k$.
This construction corresponds to 
\[
\xi_{k|k}  \boxminus \hat{\xi}_{k|k}  \sim \GP_{\text{trim}}(0,\Sigma_{k|k} )
\]
as a random variable in $\R^m$, where $\GP_{\text{trim}}$ denotes the trimmed Gaussian associated with the domain of definition of the normal coordinates. 

\subsubsection{Prediction}
\label{sec:predict}
The reference point ${\hat{\xi}_{k|k}}$ is updated using the full nonlinear model of the system
\[
\hat{\xi}_{k+1|k} = F(\hat{\xi}_{k|k}, u).
\]
The predicted error is defined as
\begin{align}
   \epsilon_{k+1|k}  := \xi_{k+1} \boxminus \hat{\xi}_{k+1|k}.
\end{align}
We use the linearisation of the state dynamics to compute an update equation for $\epsilon_{k+1|k}$ .
\begin{lemma}
The linearised dynamics of $\epsilon_{k+1|k}$ is given by
   \begin{align}
      \epsilon_{k+1|k} =  A_{k+1}\epsilon_{k|k} + \kappa_{k+1} + \Order(\lvert \epsilon_{k|k}, \kappa_{k+1} \rvert^2),
   \end{align}
where $\kappa_{k+1}\sim\GP(0,R_{k+1})$ and $A_{k+1}$ is given by
   \begin{align}
      A_{k+1} := \tD \vartheta_{\hat{\xi}_{k+1|k}}(\hat{\xi}_{k+1|k}) \cdot \tD F_{u_{k+1}}(\hat{\xi}_{k|k})\cdot \tD \vartheta^{-1}_{\hat{\xi}_{k|k}}(0).
   \end{align}
\end{lemma}
\begin{proof}
The predicted error can be written
   \begin{align}
      \epsilon_{k+1|k} &= \xi_{k+1}\boxminus\hat{\xi}_{k+1|k}, \notag  \\
      &= (F(\xi_{k}, u_{k+1})\boxplus \kappa_{k+1} )\boxminus F(\hat{\xi}_k, u_{k+1}), \notag  \\
              &= (F(\hat{\xi}_k\boxplus \epsilon_{k|k}, u_{k+1}) \boxplus \kappa_{k+1} )\boxminus F(\hat{\xi}_k, u_{k+1}), \label{eq:predicterror_dynm}
   \end{align}
Discarding $\Order( |\epsilon_{k|k}|\, |\kappa_{k+1}|)$ quadratic terms one has 
\begin{align}
\epsilon_{k+1|k}= \vartheta_{\hat{\xi}_{k+1|k}} (F_{u_{k+1}} (\vartheta^{-1}_{\hat{\xi}_{k|k}}(\epsilon_{k|k}))) + \kappa_{k+1}. \label{eq:errordynamics}
\end{align}
The formula for $A_{k+1}$ follows by applying the chain rule of differentiation and evaluating at $\epsilon_{k|k} = 0$.
\end{proof}

The predicted state error is distributed according to a Gaussian with mean zero and covariance \cite{kalman1960new}
\[
\xi_{k+1|k} \sim \GP_{\hat{\xi}_{k+1|k}}(\xi | 0,\Sigma_{k+1|k}).
\]
where 
\begin{align}
   &\Sigma_{k+1|k} = A_{k+1}\Sigma_{k|k}A_{k+1}^\top + R_{k+1}.
   \label{eq:Sigma_k+1|k}
\end{align}

This covariance update \eqref{eq:Sigma_k+1|k} hides two changes of coordinates. 
Firstly, the covariance parameter $\Sigma_{k|k}$ is defined in normal coordinates at $\hat{\xi}_{k|k}$ while $\Sigma_{k+1|k}$ is defined in normal coordinates at $\hat{\xi}_{k+1|k}$. 
Fortunately, the state update intrinsically captures this change, and $A_{k+1}\Sigma_{k|k}A_{k+1}^\top$ is the correct covariance transformation up to linearisation error.
The second change of coordinates is associated with the noise process $\kappa_{k+1}$ in \eqref{eq:predicterror_dynm} defined around the propagation of the true state $F(\xi_{k}, u)$.
However, in \eqref{eq:Sigma_k+1|k} the covariance $R_{k+1}$ is applied directly in coordinates centered at $\hat{\xi}_{k+1|k}$.
This implicit change of coordinates is further discussed in Section \ref{sec:geometric_prediction}.

\subsubsection{Update}
\label{sec:update}
The update step involves Bayesian fusion of the prior $\xi_{k+1}\sim\GP_{\hat{\xi}_{k+1|k}}(\xi| 0,\Sigma_{k+1|k})$ with a measurement $y_{k+1}\in\calN$ associated with a generative noise model $y_{k+1}\sim\GP_{h(\xi_{k+1})}(y| 0, Q_{k+1})$  \eqref{eq:configuration}.

Let $\hat{y}_{k+1|k} = h(\hat{\xi}_{k+1|k})$.
Define the innovation as $\tilde{y}_{k+1} := y_{k+1} \boxminus \hat{y}_{k+1|k}$.
The estimated measurement $\hat{y}_{k+1}$ is considered as an independent signal and $h(\xi_{k+1})$ is considered a function of the error state $\epsilon_{k+1|k}$.
The update equation uses the linearisation of the innovation with respect to $\epsilon_{k+1|k}$ at zero. 

\begin{lemma}
   The linearisation of the innovation is given by 
   \begin{align}
      \tilde{y}_{k+1} = C_{k+1} \epsilon_{k+1|k} + \nu_{k+1} + \Order(\lvert \epsilon_{k+1|k}, \nu_{k+1} \rvert^2),
   \end{align}
   where $\nu_{k+1}\sim\GP(0,Q_{k+1})$ and $C_{k+1}$, termed the \emph{output matrix}, is given by
   \begin{align}
      C_{k+1} = \tD\varphi_{\hat{y}_{k+1}}(\hat{y}_{k+1}) \cdot \tD h(\hat{\xi}_{k+1|k}) \cdot \tD \vartheta^{-1} _{\hat{\xi}_{k+1|k}}(0).
   \end{align}
\end{lemma}

\begin{proof}
One has 
   \begin{align}
      \tilde{y}_{k+1} & := y_{k+1} \boxminus \hat{y}_{k+1|k} \notag \\ 
      & = (h(\xi_{k+1}) \boxplus \nu_{k+1} )\boxminus \hat{y}_{k+1|k} \notag \\
      & = (h(\hat{\xi}_{k+1|k} \boxplus (\xi_{k+1} \boxminus \hat{\xi}_{k+1|k}) ) \boxplus \nu_{k+1} )\boxminus \hat{y}_{k+1|k} \notag \\
      & = (h(\hat{\xi}_{k+1|k} \boxplus \epsilon_{k+1|k}) \boxplus \nu_{k+1} )\boxminus \hat{y}_{k+1|k}, \label{eq:base_point_change}
   \end{align}
Discarding $\Order(|\epsilon_{k+1|k}| \, |\nu_{k+1}|)$ quadratic terms  yields
   \begin{align}\label{eq:outputerror}
      \tilde{y}_{k+1} = \varphi_{h(\hat{\xi}_{k+1|k})}(h(\vartheta^{-1}_{\hat{\xi}_{k+1|k}}\epsilon_{k+1|k})) + \nu_{k+1}.
   \end{align}
   Linearising at $\epsilon_{k+1|k} = 0$ and applying the chain rule yields   $C_{k+1}$
\end{proof} 

The Kalman update \cite{kalman1960new} for the \emph{a-posteriori} distribution is given by
\begin{align}\label{eq:update_step_result}
      \xi_{k+1|k+1} \sim \GP_{\hat{\xi}_{k+1|k}} (\xi|\mu_{k+1}, \Sigma_{k+1|k+1}) 
\end{align}
where 
\begin{align}
   K_{k+1}     &= \Sigma_{k+1|k} \, C_{k+1}^\top \, (C_{k+1}\Sigma_{k+1|k}C_{k+1}^\top + Q_{k+1})^{-1}, \label{eq:kalman_gain}\\
   \mu_{k+1}   &= K_{k+1}\tilde{y}_{k+1}, \label{eq:meanupdate}\\
   \Sigma_{k+1|k+1} &= (I - K_{k+1}C_{k+1})\Sigma_{k+1|k}.
\end{align}

Similarly to the prediction step, this update step hides an additional change of coordinate.
The generative noise process is defined around the true output $y_{k+1} = h(\xi_{k+1})$, but it is applied around the estimated output $\hat{y}_{k+1} = h(\hat{\xi}_{k+1})$.
We discuss how one can address this using parallel transport in Section \ref{sec:geoupdate}.

%

\subsubsection{Reset}
The reset step in an extended Kalman filter is geometrically a change of coordinates.
One sets 
\[
\hat{\xi}_{k+1|k+1} = \hat{\xi}_{k+1|k} \boxplus \mu_{k+1}. 
\]
The resulting \emph{posteriori} distribution is taken to be  
\[
\xi_{k+1|k+1} = \GP_{\hat{\xi}_{k+1|k+1}} (\xi|0,\Sigma_{k+1|k+1}). 
\] 
However, since the covariance $\Sigma_{k+1|k+1}$ was defined in coordinates centered at $\hat{\xi}_{k+1|k}$ it is clear that this second statement does not follow directly. 
We address this in Section \ref{sec:reset}

\section{GEOMETRIC INSIGHT}
\label{sec:geometric}
In this section, we integrate the geometric perspective into the EKF design methodology presented in Section \ref{sec:algorithm}, and propose novel geometric modifications in the filter dynamics.

\subsection{Geometric Prediction}
\label{sec:geometric_prediction}

In the prediction step \eqref{eq:predicterror_dynm}, linearisation of the error dynamics implicitly relocates the noise process $\kappa_{k+1}$ from the true state prediction $F(\xi_k, u_{k+1})$ to the estimated state prediction $F(\hat{\xi}_{k|k}, u_{k+1})$.
To correct this, the covariance prediction \eqref{eq:Sigma_k+1|k} should be modified to
\[
\Sigma_{k+1|k} = A_{k+1}\Sigma_{k|k}A_{k+1}^\top + R^+_{k+1},
\]
where $R_{k+1}^+$ is a solution to the distribution transformation
\begin{align*}
   \GP_{F(\xi_{k},u)}(\xi|0, R_{k+1})
   &\approx \GP_{F(\hat{\xi}_{k|k}, u)}(\xi|F(\xi_{k},u)\boxminus F(\hat{\xi}_{k|k}, u), R^+_{k+1}).
\end{align*}
However, solving for $R^+_{k+1}$ requires knowledge of $F(\xi_{k},u)$, and by extension $\xi_k$, which is unavailable in practice.
We do not address this issue in the present paper.

\subsection{Geometric Update}
\label{sec:geoupdate}

\begin{figure}[!htb]
   \centering
   \includegraphics[width=0.5\linewidth]{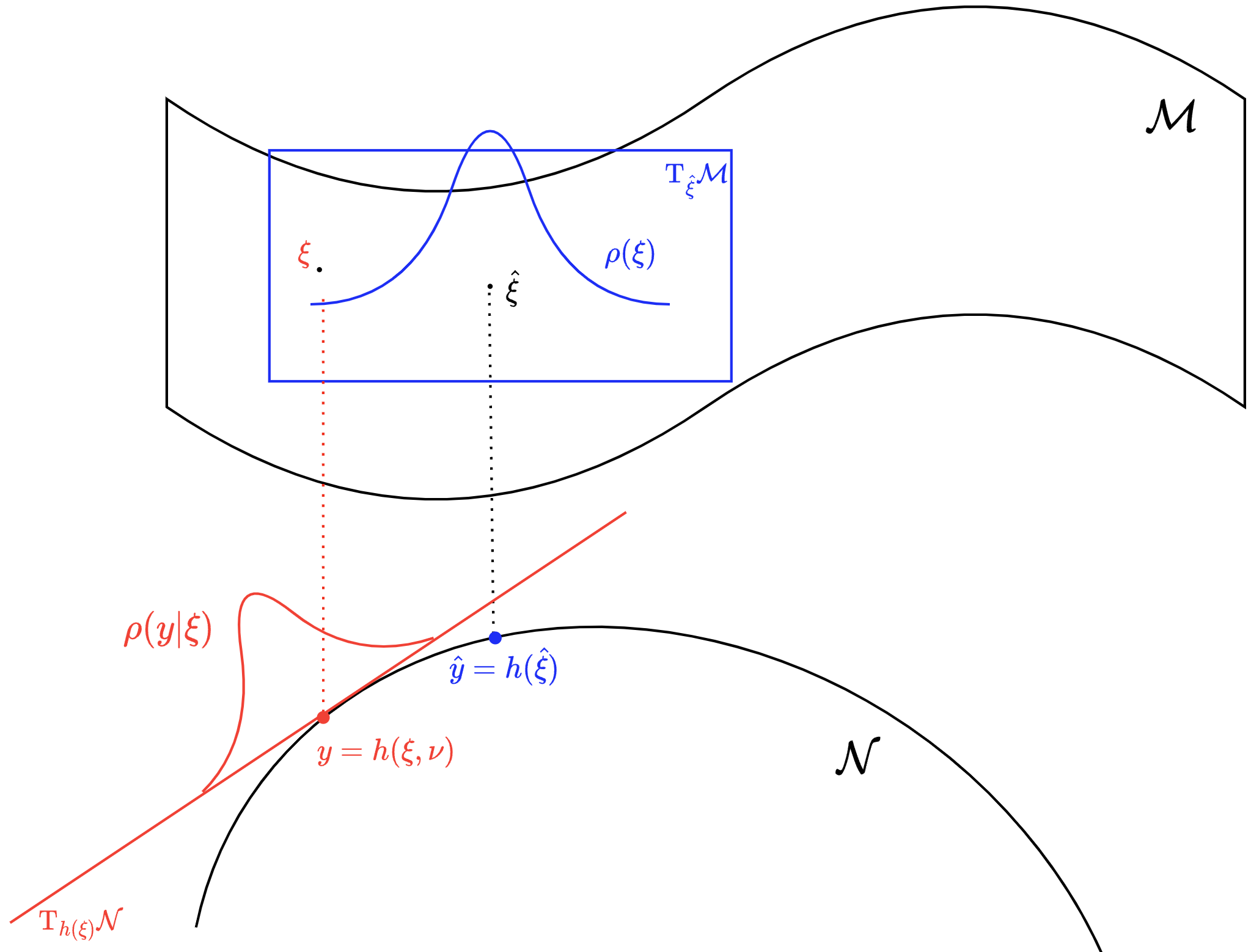}
   \caption{Demonstration of prior and measurement likelihood in normal coordinates. The prior $p(\xi_{k+1})$ is defined on the normal coordinate around ${\hat{\xi}_{k+1|k}}$, shown in blue. The measurement likelihood $p(y_{k+1}|\xi_{k+1})$ is defined on the normal coordinate ${h(\xi_{k+1})}$, shown in red.}
   \label{fig:bayes}
\end{figure}

As stated in Section \ref{sec:update}, the measurement fusion problem is solved by linearising the innovation $y_{k+1} \boxminus \hat{y}_{k+1 | k}$ in terms of $\varepsilon_{k+1 | k}$ around zero.
The covariance $C_{k+1}\Sigma_{k+1|k}C_{k+1}^\top$ associated with the filter state estimate correctly captures the change of coordinates from those centered at $\hat{\xi}_{k+1|k}$ to those centered at $\hat{y}_{k+1|k} = h(\hat{\xi}_{k+1|k})$, up to linearisation error. 
However, the linearisation \eqref{eq:base_point_change} also implicitly changes the base point of the noise process $\nu_{k+1}$ from the true output $h(\xi_{k+1})$ to the estimated output $h(\hat{\xi}_{k+1})$, introducing error into the Kalman gain computation \eqref{eq:kalman_gain}.
To address this, we propose a novel geometric update step in the filter design, which translates the likelihood distribution from $\tT_{h(\xi_{k+1})}\calN$ to $\tT_{\hat{y}_{k+1}}\calN$ in a manner compatible with the affine connection of $\calN$.
Our goal is to approximate
\begin{align}
   y_{k+1} \sim \GP_{h(\xi_{k+1})}(y|0, Q)
   &\approx \GP_{\hat{y}_{k+1}}(y|h(\xi_{k+1})\boxminus\hat{y}_{k+1}, Q^+),
\end{align}
where $Q^+$ denotes the transformation of the original measurement covariance $Q$ to $\tT_{\hat{y}_{k+1}}\calN$.
By modelling $Q$ as a (2,0)-tensor on $\tT_{h(\xi_{k+1})}\calN$, it follows that
\begin{align}\label{eq:transport_covariance_update}
   Q^+ = \PT_{\gamma(1)} Q,
\end{align}
where $\gamma(t) = h(\xi_{k+1}) \boxplus t(\hat{y}_{k+1|k} \boxminus h(\xi_{k+1}))$ is the geodesic from the true output $h(\xi_{k+1})$ to the estimated output $\hat{y}_{k+1}$.
With this new covariance $Q^+$, the standard filter update step can be applied.



The parallel transport \eqref{eq:transport_covariance_update} requires knowledge of the true output $h(\xi_{k+1})$ which is not available in practice.
Here we propose two choices of approximation.
\subsubsection{Measurement $y_{k+1}$}
The measurement $y_{k+1}$ can be used as an approximation of $h(\xi_{k+1})$.
As shown in Fig \ref{fig:bayes}, the likelihood distribution is centred at $h(\xi_{k+1})$.
When the state estimate is poor relative to the measurement, one may take the likelihood to be centred at $y_{k+1}$, and thus approximate the parallel transport from $h(\xi_{k+1})$ to $\hat{y}_{k+1}$ by the one from $y_{k+1}$ to $\hat{y}_{k+1}$.

\subsubsection{Naive posteriori}\label{sec:naiveposter}
Another alternative is given by the posterior estimated by the original filter update; that is, without geometric update.
Following the process in Section \ref{sec:update}, one can compute an estimate of the posterior $\hat{\xi}_{k+1|k+1}$.
This new estimate can be used to approximate the true output $y_{k+1}$ and repeat the update step to perform the geometric update.
This process of estimating a more accurate posterior can be iterated to achieve better performance, at a higher computational cost.
Algorithm \ref{alg:iter} provides pseudocode for this iterated update.

\begin{remark}
   Some existing filter designs, such as the iterated Kalman filter \cite{sorenson1966kalman} and the iterated EKF on Lie groups \cite{bourmaud2016intrinsic}, use an iterative scheme to generate a more accurate estimation of the output matrix $C_{k+1}$ than can be obtained through a single linearisation step.
   In contrast, Algorithm \ref{alg:iter} uses iteration to better approximate the measurement noise covariance $Q_{k+1}$ in the correct coordinates, while the output matrix is computed in a single step.
   \end{remark}

\begin{algorithm}[htb]
   \caption{Iterated geometric update in the proposed EKF}\label{alg:iter}
   \textbf{Input:} prior $(\hat{\xi}_{k+1|k}, \Sigma_{k+1|k})$, likelihood $Q_{k+1}$\\
   \textbf{Input:} measurement $y_{k+1}$\\
   $Q^0_{k+1} \gets Q_{k+1},$\\
   \For{\texttt{i in range (num of iter)}}{

   \begin{align*}
      &K'_{k+1}     = \Sigma_{k+1|k} \, C_{k+1}^\top \, (C_{k+1}\Sigma_{k+1|k}C_{k+1}^\top + Q^i_{k+1})^{-1},\\
      &\mu'_{k+1}   = K'_{k+1}\tilde{y}_{k+1},\\
      &\hat{\xi}_{k+1|k+1}' = \hat{\xi}_{k+1|k}\boxplus \mu'_{k+1} ,\\
      &\bar{y}'_{k+1} = h(\hat{\xi}_{k+1|k+1}'),\\
      &Q^{i+1}_{k+1} = \PT_{\gamma(t)} Q_{k+1}\; \text{where}\; \gamma(t) = \bar{y}'_{k+1}\boxplus t(\hat{y}_{k+1}\boxminus \bar{y}'_{k+1}),
   \end{align*}
   }
   \textbf{Update:}{
   \begin{align*}
      &K_{k+1}  = \Sigma_{k+1|k} \, C_{k+1}^\top \, (C_{k+1}\Sigma_{k+1|k}C_{k+1}^\top + Q^{i+1}_{k+1})^{-1},\\
      &\mu_{k+1} = K_{k+1}\tilde{y}_{k+1},\\
      &\Sigma_{k+1|k+1} = (I - K_{k+1}C_{k+1})\Sigma_{k+1|k}.
   \end{align*}
   }
\end{algorithm}


\subsection{Geometric Reset}
\label{sec:reset}


There have been several works on the covariance reset in error-state Kalman filters.
It was first mentioned by Markley \cite{markley2003attitude} in the context of multiplicative EKF, recently generalised by Muller \etal \cite{mueller2017covariance}\cite{gill2020full}.
The authors proposed the covariance reset step using parallel transport in \cite{Mahony_2022}\cite{ge2022equivariant} for filtering on homogeneous spaces.
The same concept can be extended onto a smooth manifold.

At the end of the update step, the posterior distribution \eqref{eq:update_step_result} is a concentrated Gaussian about the predicted state $\hat{\xi}_{k+1|k}$, but with a non-zero mean $\mu_{k+1}$ and updated covariance $\Sigma_{k+1|k+1}$.
However, the next filter iteration requires that the state estimate is expressed in coordinate centred at $\hat{\xi}_{k+1|k+1} = \hat{\xi}_{k+1|k} \boxplus \mu_{k+1}$ so that the mean of the distribution is zero.
The goal of the reset step is to identify $\Sigma^+_{k+1|k+1}$ such that
\begin{align}\label{eq:reset}
   \xi_{k+1} \sim \GP_{\hat{\xi}_{k+1|k}}(\xi|\mu_{k+1},\Sigma_{k+1|k+1})\approx\GP_{\hat{\xi}_{k+1|k+1}}(\xi|0,\Sigma^+_{k+1|k+1}).
\end{align}
Similarly to the geometric update, this may be solved using parallel transport on $\calM$.
The covariance $\Sigma$ is modelled as a (2,0)-tensor on $\tT\calM$.
Then the reset covariance $\Sigma^+_{k+1|k+1}$ is found to be
\begin{align}
\Sigma^+_{k+1|k+1} & = \PT_{\gamma(1)} \Sigma_{k+1|k+1} \label{eq:Sigma_reset},
\end{align}
where $\gamma(t)$ is the geodesic curve $\gamma(t) = \hat{\xi}_{k+1|k}\boxplus t\mu_{k+1}$.
Unlike the geometric update step, all the information required to solve \eqref{eq:Sigma_reset} exactly is available in practice.


\section{SIMULATION}
\label{sec:sim}
In this section, we use an example of attitude estimation from two directional measurements to demonstrate the algorithm performance.
The results of different choices of reference are compared.

\subsection{System Formulation}
We consider an attitude estimation problem with gyroscope input and measurements of two known directions.
Let \{$G$\} and \{$I$\} denote the global reference frame and the body-fixed reference frame, respectively.
Let $^GR\in\SO(3)$ denote the rigid body orientation of a moving rigid platform.
The onboard gyroscope, which has the same orientation as the platform, returns the bias-free angular velocity $^I\omega\in\R^3$.
With non-rotating, flat earth assumption, the deterministic system kinematics is given by
\begin{align}\label{eq:attitude}
   \idx{{R}}{G}{}{k+1} = \idx{R}{G}{}{k} \exp \left(  \ids{\omega}{I}{}{}{\wedge}\delta t\right).
\end{align}

For the estimation problem, we consider the measurements of two known directions $d_1, d_2$.
The output space is then defined to be $\calN:=\mathrm{S}^2\times \mathrm{S}^2$.
The configuration output is written
\begin{align}\label{eq:output}
   h({}^GR) = ({}^GR^\top d_1, {}^GR^\top d_2),
\end{align}
where $d_1$ and $d_2$ satisfy $d_1 \times d_2 \neq 0$; i.e. the state is always observable.

\subsection{Implementation}\label{sec:implementation}
We simulate an oscillatory trajectory for attitude estimation \eqref{eq:attitude}.
The state ${}^G{R}$ is initialized with identity rotation matrix.
The angular velocity input is defined to be $^I\omega = (0.1\times\cos(\tau), 0.1\times\sin(\tau), 0.1\times\sin(\tau))$~rad/s.
The trajectory is realized using Euler integration at time step $\delta t=0.02s$.
The estimator has an onboard gyroscope which reads the angular velocity but is corrupted by piecewise constant zero-mean white Gaussian noise with variance 0.02~(rad/s)$^2$ per axis.
Additionally, there are sensors that provide measurements of two known directions $d_1 = (0,1,0)$ and $d_2 = (1/\sqrt{2},0,1/\sqrt{2})$.
The directional measurement \eqref{eq:output} are corrupted by Gaussian noise with zero-mean and non-homogeneous covariance $\diag(0.01, 0.03, 0.05)$~rad$^2$.

We implement the geometric extended Kalman filter described in Section \ref{sec:algorithm}, where the manifold $\SO(3)$ is equipped with the Cartan-Schouten 0-connection \cite{nomizu1954invariant}.
Consequently, the normal coordinates are exactly the Lie group exponential coordinates, and the extended Kalman filter is equivalent to the Equivariant filter described in \cite{fornasier2022overcoming}.
The parallel transport on the state space $\SO(3)$ is solved using the Cartan-Shouten 0-connection on Lie groups.
On the output space $\mathrm{S}^2\times \mathrm{S}^2$, we use the canonical affine connection induced by the Cartan-Shouten 0-connection on the Symmetry Lie group $\SO(3)$ \cite{nomizu1954invariant}.
For comparison, we study the proposed filter with and without our proposed geometric update and covariance reset.
All filters are initialized by sampling the concentrated Gaussian distribution $\hat{R}_{0} \sim \GP_{I_3}(0, 1.5^2 I_3)$.

\subsection{Simulation Results}

To compare their performance, we plot the error in the attitude estimate as well as the filter energy.
The filter energy is $\frac{1}{m}\epsilon^\top\Sigma^{-1}\epsilon$ with $m$ being the dimension of the system state and follows a $\chi^2$ distribution.
The expected value of filter energy is 1, while a smaller or larger value indicates the filter is under-confident or over-confident in its estimation, respectively.

\begin{figure}[!htb]
   \centering
   \includegraphics[width=\linewidth]{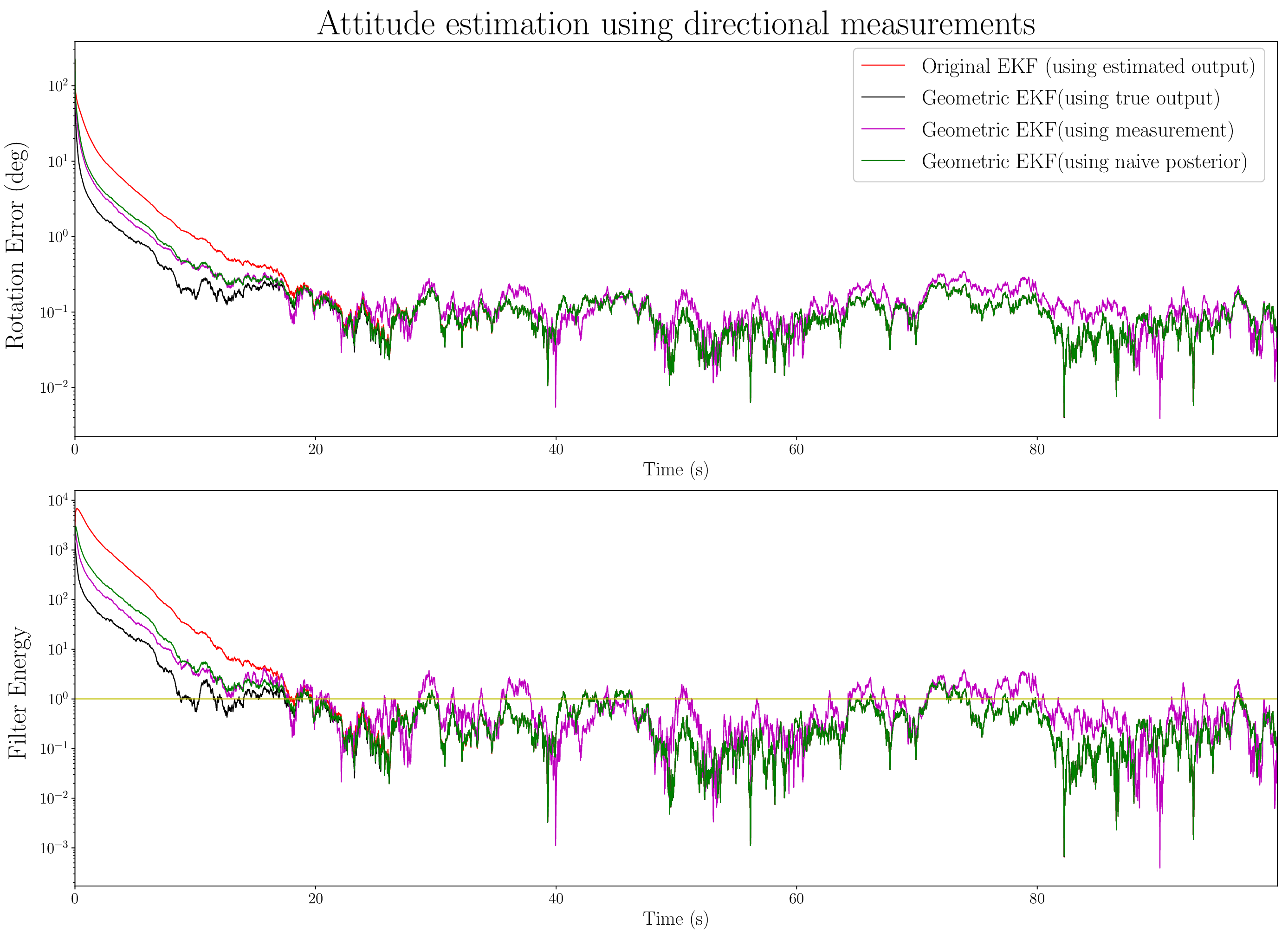}
   \caption{
      The rotation estimation error and filter energy are shown for various EKF implementations.
      The original EKF implementation ({\color{red}---}) discussed in Section \ref{sec:algorithm} is compared with the EKF with covariance reset and geometric update obtained from the true output ({\color{black}---}), the measured output ({\color{magenta}---}), and the naive posterior ({\color{blue}---}).
      The yellow horizontal line in the second subplot is for filter energy 1.
   }
   \label{fig:result_1}
\end{figure}

Figure \ref{fig:result_1} shows the performance of the filter with and without the proposed geometric modifications.
Due to the large initial error in the state estimate, relative to the measurements, the proposed geometric update is seen to make a significant improvement to the convergence of the filter in the transient period.
The black trajectory is obtained by using the true output for parallel transport in the geometric update, and is presented only as a reference since this is not possible in practice.
It has the best performance over all four implementations.
The green and purple trajectories, which are using the naive posterior and measurement for parallel transport, respectively, perform similarly in the transient period, and both outperform the conventional EKF implementation (in red).
In terms of the asymptotic behavior, however, the EKF using measurements for parallel transport generally performs worse than all alternatives.
The reason is that once the filter is converged, the state estimate tends to be more accurate than the measurement in approximating the true output.
In this case, using the measurement as an approximation of the true output will only decrease filter performance.

\begin{figure}[!htb]
   \centering
   \includegraphics[width=\linewidth]{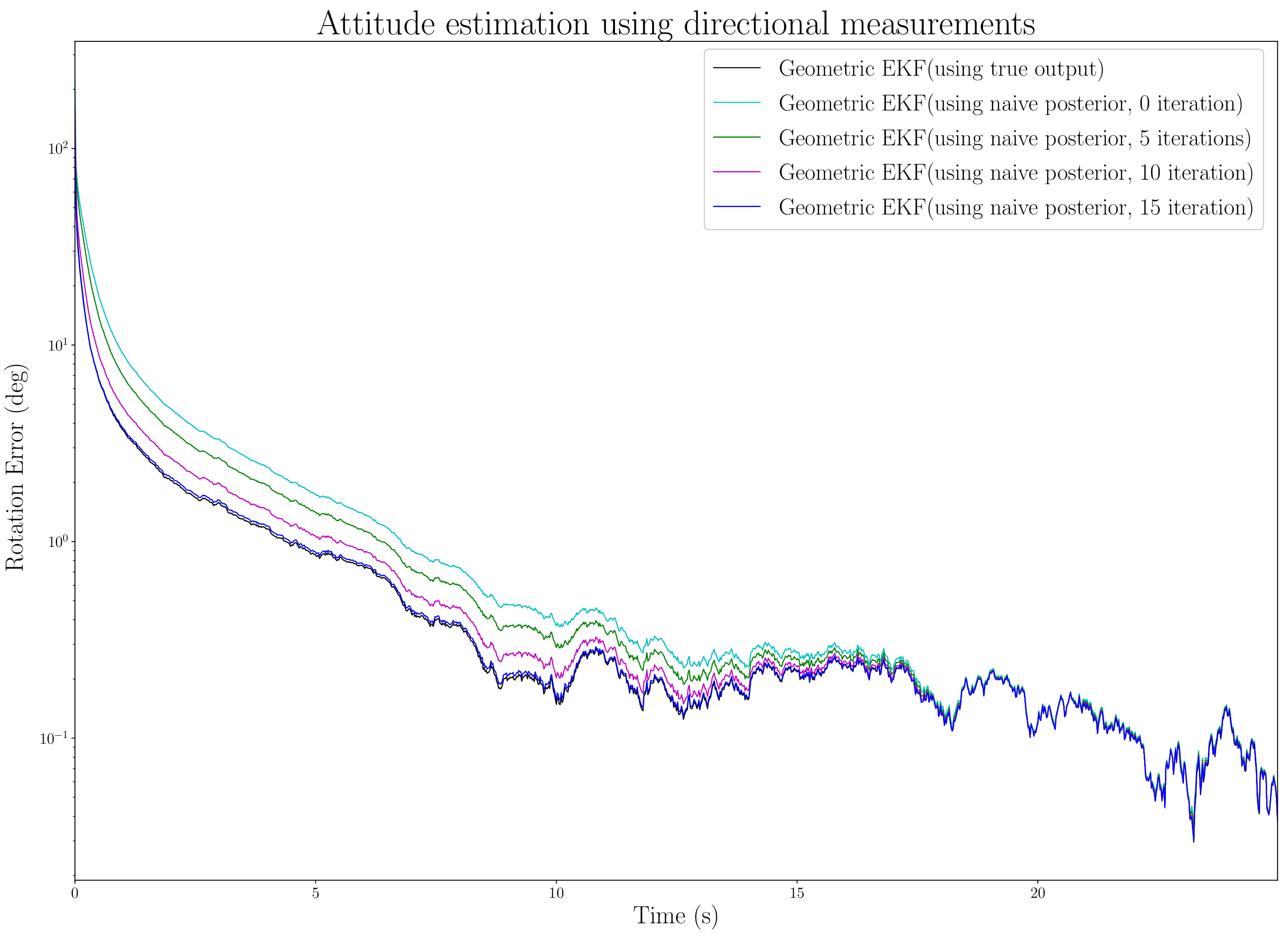}
   \caption{Comparison of the performance of EKF implementations with geometric update using 0 ({\color{cyan}---}), 5 ({\color{ForestGreen}---}), 10 ({\color{magenta}---}), and 15 ({\color{blue}---}) iterations.
   An EKF using the true output for the geometric update ({\color{black}---}) is included as a reference for comparison.
   }
   \label{fig:result_iter}
\end{figure}

Figure \ref{fig:result_iter} shows the results of EKF implementations with iterated geometric updates discussed in Algorithm \ref{alg:iter}.
We implement the EKF with 0, 5, 10, and 15 iterations in the geometric update.
As in Figure \ref{fig:result_1}, an EKF where the true output is used for the geometric update is provided as a reference for comparison.
One observes that the performance of the EKF with iterated update improves as more iterations are used, and approaches that of the EKF with true output update.
Each iteration requires additional computation and this must be balanced with improved accuracy; however, these results clearly demonstrate the advantage of including the proposed geometric update in an EKF design.

\section{CONCLUSIONS}
\label{sec:conclusion}

This paper presents an error-state extended Kalman filter design methodology for smooth manifolds with affine connections.
The proposed algorithm includes additional geometric modifications in the filter update and error reset steps by applying parallel transport to the state and measurement covariance matrices.
The theory is applied to an example problem of attitude estimation with two directional measurements.
The simulation results demonstrate the convergence of the proposed EKF, and show the improvements in performance gained from applying the proposed geometric modifications.

\bibliography{reference}
\bibliographystyle{IEEEtran}

\end{document}